\documentclass[11pt,letterpaper]{article}

\usepackage[margin=1in]{geometry}
\usepackage{amsmath,amssymb,amsthm}
\usepackage{algorithm}
\usepackage{algpseudocode}
\usepackage{mathtools}
\usepackage{enumitem}
\usepackage{microtype}
\usepackage{tikz}
\usepackage{xspace}
\usepackage{hyperref}

\usetikzlibrary{arrows.meta,positioning,fit,calc,backgrounds,decorations.pathreplacing}

\definecolor{introblue}{RGB}{55,126,184}
\definecolor{tokorange}{RGB}{230,159,0}
\definecolor{exchgreen}{RGB}{0,158,115}

\hypersetup{
  colorlinks=true,
  linkcolor=blue,
  citecolor=blue,
  urlcolor=blue
}

\newtheorem{theorem}{Theorem}
\newtheorem{lemma}[theorem]{Lemma}

\newtheorem{corollary}[theorem]{Corollary}
\theoremstyle{definition}
\newtheorem{definition}[theorem]{Definition}
\newtheorem{example}[theorem]{Example}
\newtheorem{remark}[theorem]{Remark}

\newcommand{\OPT}{\mathrm{OPT}}
\newcommand{\cost}{\mathrm{Cost}}
\newcommand{\E}{\mathbb{E}}
\newcommand{\Prb}{\mathbb{P}}
\newcommand{\1}{\mathbf{1}}

\newcommand{\Z}{Z}
\newcommand{\tw}{\tilde w}

\algrenewcommand\algorithmicrequire{\textbf{Input:}}
\algrenewcommand\algorithmicensure{\textbf{Output:}}
\algnewcommand{\algorithmicfixed}{\textbf{Fixed parameters:}}
\algnewcommand{\Fixed}{\item[\algorithmicfixed]}
\algtext*{EndFor}
\algtext*{EndIf}

\title{Transposition Is Nearly Optimal for IID List Update}
\author{
Christian Coester\\
University of Oxford
}
\date{}

\begin{document}
\maketitle

\begin{abstract}
The list update problem is one of the oldest and simplest problems in online algorithms: A set of items must be maintained in a list while requests to these items arrive over time. Whenever an item is requested, the algorithm pays a cost equal to the position of the item in the list. In the i.i.d. model, where requests are drawn independently from a fixed distribution, the static ordering by decreasing access probabilities $p_1\ge p_2\ge \dots \ge p_n$ achieves the minimal expected access cost $\OPT=\sum_{i=1}^n ip_i$. However, $p$ is typically unknown, and approximating it by tracking access frequencies creates undesirable overheads.

A 50-year old conjecture by Rivest states that the \emph{Transposition} rule (swap the requested item with its predecessor) is optimal among memoryless algorithms. While known to be false in the strict sense, prior research provided \emph{asymptotic} confirmation of Rivest's conjecture for restricted distributions. We show that for \emph{all} distributions, Transposition has stationary expected access cost at most $\OPT+1$, confirming Rivest's conjecture up to an unavoidable additive constant, and showing near-optimality even among algorithms \emph{with memory}. More abstractly, it yields a purely memoryless procedure to approximately sort probabilities via sampling.
\end{abstract}

\section{Introduction}

Maintaining frequently accessed objects near the front of a list is a ubiquitous primitive: linked-list dictionaries, symbol tables, rule lists, and hash-table buckets all reduce to repeatedly searching a linear list and occasionally reorganizing it. The \emph{list update} problem formalizes this: $n$ items must be arranged in a list. At each time step, one of the items is requested, incurring cost equal to the position of the item in the list. After each request, the algorithm is allowed to rearrange the list. Moving the requested item closer to the front of the list is usually considered free, while any other exchanges would incur a cost of 1.

The problem has been studied for over six decades, going back to McCabe~\cite{McCabe65}. Most of the early work on the problem considered the \emph{i.i.d. model}, where requests at different time steps are sampled independently from the same distribution. Later, it was also studied in the adversarial setting, and became foundational to the area of competitive analysis~\cite{SleatorT85}. We focus here on the original i.i.d. setting~\cite{McCabe65}. Our main result is that the Transposition rule (swap the requested item with its predecessor) is nearly optimal, confirming a 50-year old conjecture by Rivest~\cite{Rivest76} up to an unavoidable additive constant.

Denote by $p_i$ the probability of requesting item $i$, and assume without loss of generality that items are labeled such that $p_1\ge p_2\ge\dots\ge p_n$. Clearly, the optimal policy is to maintain the static ordering by decreasing probabilities, which achieves the minimal expected access cost of
\begin{align*}
\OPT=\sum_{i=1}^n i p_i.    
\end{align*}
This strategy presumes exact knowledge of the access probabilities, which is typically unavailable. A natural alternative is to order items by their observed request frequencies, which converges asymptotically to the same expected cost. However, as noted by Knuth \cite{Knuth98} and Rivest~\cite{Rivest76}, the extra memory space required for maintaining frequency counters would be undesirable as it could be better used for alternative nonsequential search techniques.

Research has therefore focused on \emph{memoryless (self-organizing)} heuristics, which do not maintain any memory beyond the current ordering of items. The two rules that have been studied by far the most extensively since the 1960s are \emph{Move-to-Front} and \emph{Transposition}~\cite{McCabe65}:

\begin{description}
    \item[Move-to-Front:] Move the requested item to the front of the list, without changing the relative order of other items.
    \item[Transposition:] Swap the requested item with its immediate predecessor (unless it is already at the front).
\end{description}

If requests are adversarial, then Move-to-Front is known to be superior, achieving a constant competitive ratio whereas the competitive ratio of Transposition is unbounded \cite{SleatorT85}. Indeed, consider alternately requesting the last two items of the list; then Transposition incurs a cost of $n$ per time step, whereas an optimal algorithm could bring those two items to the front. Move-to-Front is also known to perform better on request sequences exhibiting strong locality of reference~\cite{HesterH85,AngelopoulosDL08,AlbersL16}.

The situation reverses in the i.i.d. setting: For any distribution $p$, any memoryless rule induces a Markov chain whose states are the permutations of the items. The rules are typically evaluated by their \emph{stationary} expected cost, capturing long-run behavior when the access law is stable. For Move-to-Front, \cite{ChungHS88} showed that its stationary expected cost is within a factor $\frac{\pi}{2}\approx 1.57$ of $\OPT$ for any distribution $p$, and this is tight when $p_i\propto 1/i^2$~\cite{GonnetMS81}. McCabe~\cite{McCabe65} conjectured that Transposition has stationary expected cost less than or equal to Move-to-Front. This was later proved by Rivest~\cite{Rivest76}, who further showed that equality holds only in trivial cases ($n\le 2$ or $p_i=1/n$). Thus, Transposition is preferable in the i.i.d. setting, or if the distribution evolves only slowly. Transposition has also been noted for its efficient implementation using arrays, since swapping adjacent entries is very fast. This avoids the overhead of maintaining pointers as needed for linked list implementations, which are typically used for Move-to-Front~\cite{HesterH85}.

A 50-year-old conjecture by Rivest~\cite{Rivest76} states that Transposition is optimal among \emph{memoryless} rules for \emph{any} distribution $p$. Supporting this, Andrew Yao showed that there exist distributions where Transposition is the unique optimal memoryless rule (see~\cite{Rivest76}). Six years later, the literal version of Rivest's conjecture was shown to be false by a counterexample with $n=6$ items~\cite{AndersonNW82}. Focus therefore shifted to proving \emph{asymptotic} versions of Rivest's conjecture. However, the question turned out to resist progress on the case of general distributions for several decades, with advances instead restricted to specific distributions: \cite{KanR80,PhelpsL80,Makjamroen92} obtained positive answers for the case where all but one item have the same access probability. \cite{GamarnikM05}  studied two families of distributions and showed asymptotic optimality of the logarithm of tail probabilities of the search cost. More recently, Indyk, Quaye, Rubinfeld and Silwal \cite{IndykQRS25} showed for Zipfian distributions $p_i\propto 1/i^\alpha$ with $0<\alpha\le 2$ that the stationary expected cost of Transposition is at most $(1+o(1))\OPT$ as $n\to\infty$. Note that $n\to\infty$ alone does not suffice to achieve ratio $1+o(1)$ in general, as the example in~\cite{AndersonNW82} can be extended to arbitrary $n$ by adding items with vanishing probability. For general distributions, prior to our work the best upper bound on the cost of Transposition remained $\frac{\pi}{2}\OPT$, inherited from the according tight bound on Move-to-Front and the aforementioned property that Transposition is no worse than Move-to-Front.

An intuitive reason why Transposition is harder to analyze than Move-to-Front is the following: In Move-to-Front, the list at a given time has a simple structure, ordered by last request times of the items. In contrast, under Transposition the position of an item is shaped by a more intricate interaction of past requests. The resulting difficulty of analyzing Transposition has been noted explicitly, for example, in the survey of Hester and Hirschberg \cite{HesterH85}:
\begin{quote}
    ``Transpose is widely mentioned in the literature, but authors merely state an inability to obtain theoretical bounds on its performance.''
\end{quote}

The main result of this paper is that Transposition is optimal up to a constant additive term, for \emph{any} distribution:

\begin{theorem}\label{thm:main}
In the i.i.d. model, the stationary expected cost of Transposition is at most $\OPT+1$.
\end{theorem}

This confirms Rivest's original conjecture up to an additive constant (which is unavoidable due to the example of~\cite{AndersonNW82}). Note that the $+1$ term is multiplicatively negligible whenever the probability mass is sufficiently spread across a superconstant number of items. In particular, the aforementioned $1+o(1)$ multiplicative guarantee of \cite{IndykQRS25} for Zipfian distributions follows immediately from Theorem~\ref{thm:main}, since $\OPT\to\infty$ for the distributions considered there.

While Rivest's original conjecture concerns the performance of Transposition relative to other memoryless rules, Theorem~\ref{thm:main} shows that Transposition is nearly optimal even compared to algorithms \emph{with} memory. In other words, a distribution-oblivious local update rule matches the distribution-aware optimum up to an additive constant, despite maintaining no explicit estimates of the access law. Since the distribution-aware optimum corresponds to the decreasing probability order, an alternative interpretation of Theorem~\ref{thm:main} is that Transposition can be used as a method to approximately \emph{sort} the elements of any probability distribution over a finite domain. Remarkably, this only requires sample access to the distribution, and yields an in-place algorithm without maintaining additional memory such as frequency estimates.

The stationary distribution of Transposition is known to be identical to that of the \emph{gladiator chain}~\cite{HaddadanW19}, a Markov chain on the symmetric group $S_n$ where each element $i\in[n]$ can be thought of as a gladiator with some strength $p_i>0$ (not necessarily summing to $1$), and a permutation $\pi\in S_n$ maps gladiators to ranks: in each step, for $k\in[n-1]$ uniformly at random, gladiators $i=\pi^{-1}(k)$ and $j=\pi^{-1}(k+1)$ with adjacent ranks fight over their position in the ranking.\footnote{Similar dynamics are used to determine rankings in ``bumps'' rowing races.} Gladiator $j$ wins and moves ahead of gladiator $i$ with probability $\frac{p_j}{p_i+p_j}$ (also known as the Bradley-Terry model~\cite{BradleyT52} and studied already by Zermelo in the 1920s~\cite{Zermelo29}). Our proof implies a quantitative guarantee on the stationary ranking: for each gladiator $j$, the expected total excess strength of stronger gladiators ranked below $j$ is at most $p_j$.

\begin{corollary}\label{cor:gladiator}
    In the stationary distribution of the gladiator chain, for all $j\in[n]$,
    \[
    \mathbb{E}\left[\sum_{i\colon \pi(i)>\pi(j)}\max\{p_i-p_j, 0\}\right]\le p_j.
    \]
\end{corollary}
In particular, gladiators do not outrank substantially stronger opponents in expectation.

Invoking a recent lower bound on the spectral gap of biased adjacent-transposition chains~\cite{GreavesZ26}, we obtain that, at least when $p_i\ge n^{-O(1)}$, the cost bound of Theorem~\ref{thm:main} holds not only in stationarity, but already after polynomially many steps:

\begin{theorem}
\label{thm:poly-convergence-inverse-poly}
If access probabilities in i.i.d. list update are lower bounded by $p_i\ge n^{-c}$, where $c>0$ is a fixed constant, then, from every initial configuration, after
\[
    O(n^{c+4}\log n)
\]
requests the Transposition chain is within total variation distance
\(n^{-2}\) of its stationary distribution. Consequently\footnote{The proof of Theorem~\ref{thm:main} actually yields a slightly stronger bound of $\OPT+1-\frac{1}{n}$ in stationarity.}, after that many
requests,
\[
    \mathbb{E}[\mathrm{Cost}(\pi_t)]
    \le
    \OPT+1 .
\]
\end{theorem}
Observe that a lower bound $n^{-c}$ is satisfied, for example, by Zipfian distributions.

\subsection{Additional Related Work}

Self-organizing lists in the i.i.d. setting have been the topic of several surveys~\cite{Hendricks76,HesterH85}. Besides analyzing the cost in stationarity, algorithms have also been analyzed with respect to their rate of convergence to this stationary cost. Move-to-Front is known to converge faster than Transposition, and Bitner suggests a hybrid rule that combines the fast convergence of Move-to-Front with the stationary performance of Transposition~\cite{Bitner79}. Empirical evaluations of Transposition and other rules include~\cite{Tenenbaum78,BachrachE97}. For Move-to-Front in the i.i.d. setting, prior research was also able to derive formulas for the distribution over permutations after any finite number of requests~\cite{Fill96exact}, the distribution of search cost in stationarity~\cite{FillH96distribution}, and the rate of convergence to this distribution~\cite{Fill96limits}. Among recent works, closest to ours is the aforementioned~\cite{IndykQRS25}.

In the adversarial competitive analysis setting, Move-to-Front achieves the optimal competitive ratio of $2$ among deterministic algorithms~\cite{SleatorT85}, and for randomized algorithms the competitive ratio is known to lie between $1.5$ and $1.6$~\cite{Irani91,Teia93,ReingoldWS94,Albers98,AlbersSW95,AmbuhlGS13}. In recent years, the list update problem has been studied under many variants and analysis frameworks, including paid exchanges~\cite{AlbersJ20,BasiakB0CJST25}, algorithms with predictions~\cite{AzarLS25}, advice~\cite{BoyarKLL17}, delays and time windows~\cite{AzarLV24}, itinerant list update~\cite{OlverPSSS18}, locality of reference~\cite{AngelopoulosDL08,AlbersL16}, bijective analysis~\cite{AngelopoulosS13}, and other models surveyed in~\cite{KamaliL13}. A closely related online problem is caching, corresponding to the variant of list update where the access cost is $0$ for the first $k$ items in the list and $1$ otherwise~\cite{SleatorT85}.

The gladiator chain and related Markov chains on the symmetric group $S_n$ have been studied independently, mostly with respect to their spectral gap and mixing time. If all gladiators have the same strength, the mixing time is known to be $\Theta(n^3\log n)$~\cite{Wilson04,Lacoin16} (tightening prior bounds from~\cite{DiaconisS81,DiaconisS93}), and \cite{HaddadanW19} show polynomial mixing time for the case of three strength classes. Recent work~\cite{GheissariLV26} shows rapid mixing for a more general model provided the stronger gladiator wins each comparison with probability at least $\frac{1}{2}+\epsilon$. Mixing times of closely related Markov chains were studied in~\cite{BenjaminiBHM05,BhaktaMRS13}. Fill~\cite{Fill03} conjectured that the case of equally strong gladiators minimizes the spectral gap, even among a more general class of regular biased adjacent-transposition chains. This was recently confirmed by~\cite{GreavesZ26}.

\subsection{Proof Overview}\label{sec:proofOverview}

We use the classical formula for the stationary distribution of the Transposition Markov chain, $Q(\pi)\propto \prod_{i=1}^n p_i^{n-\pi(i)}$ for $\pi\in S_n$~\cite{Rivest76,Hendricks76}. However, turning this formula into a bound on the expected cost is not straightforward as the search cost depends on many correlated inversion events.

Our starting point is a decomposition of suboptimality into inversions. For any (random) permutation $\pi$, the excess over the optimum $\OPT(p)=\sum_{i=1}^n i\,p_i$ can be written as a weighted inversion sum: each inverted pair of indices $i<j$ with $\pi(j)<\pi(i)$ contributes exactly $p_i-p_j$ to
the excess cost (Lemma~\ref{lem:inversions}).
We then regroup these contributions by the \emph{less likely} item.
For each $j\ge 2$ we define
\[
s_j := \sum_{i<j} (p_i-p_j)\Prb_{\pi\sim Q}[\pi(j)<\pi(i)],
\]
so that $\mathbb{E}_{\pi\sim Q}[\cost(\pi)]-\OPT(p)=\sum_{j=2}^n s_j$. Intuitively, $s_j$ represents the excess cost that can be blamed on item $j$ appearing ahead of higher-probability items in the permutation. As a key step of our proof, we show the per-item inequality $s_j\le p_j$. It immediately implies Theorem~\ref{thm:main}, since then $\sum_{j=2}^n s_j \le \sum_{j=2}^n p_j \le 1$.

\paragraph{Intuition for $\boldsymbol{s_j\le p_j}$.} We briefly provide some intuition for why the inequality $s_j\le p_j$ appears plausible. At a very high level, a monotone relationship between $s_j$ and $p_j$ makes sense since $s_j$ is the blame on item $j$ for overtaking higher probability items, and larger $p_j$ facilitates such overtaking. For small $n$, the inequality can be checked directly: E.g., when $n=2$, then $s_2=\frac{p_1-p_2}{p_1+p_2}p_2\le p_2$ (which is tight when $p_2\to 0$). More intuition can be derived from the case where all items other than $j$ have the same request probability $q>p_j$. Then the position of $j$ follows a biased random walk. In its stationary distribution, each step that $j$ is closer to the front is a factor $p_j/q$ less likely, and the expected number of items behind $j$ is at most $\frac{p_j}{q-p_j}$ (which is tight as $n\to\infty$). Thus, $s_j\le(q-p_j)\frac{p_j}{q-p_j}$. However, the general case is more complicated because the higher-probability items need not have the same probability $q$. Moreover, items of probability less than $p_j$ can interfere with the process, but intuitively such interference only creates buffers that make it harder for $j$ to overtake higher-probability items.

\paragraph{Proving $\boldsymbol{s_j\le p_j}$.} To prove $s_j\le p_j$, we plug the formula for $Q(\pi)$ into the definition of $s_j$, which reduces the inequality to nonnegativity of a polynomial in the variables $p_1,\dots,p_n$. A convenient reparameterization is to pass from the probabilities $(p_1,\ldots,p_n)$ to the \emph{gap variables} $x_i:=p_i-p_{i+1}$ (with $x_n:=p_n$). In these variables, the monotonicity assumption $p_1\ge \cdots \ge p_n$ becomes simply $x_i\ge 0$ for all $i$, turning the inequality into nonnegativity of a polynomial on $\mathbb R^n_{\ge 0}$.

At this point, the remaining challenge is combinatorial: we show that \emph{every coefficient} of the resulting polynomial is nonnegative. To do this, we interpret each coefficient as the difference $|B|-|A|$ of the cardinalities of two sets. Here, $B$ is a set of \emph{word tuples}, and $A$ is a set of (word tuple, index) pairs. We certify nonnegativity by giving an injection from $A$ to $B$ (Algorithm~\ref{alg:injection}).

Constructing this map in a way that is, on the one hand, invertible, and on the other hand, respects the constraints of the sets $A$ and $B$, is where most of the proof complexity resides. Our solution involves shortening one of the input words, and using most of its removed letters plus one special ``deficit'' letter to either insert them directly into ``receiver'' words or to replace a letter in an ``exchange'' word and add the replaced letter to a receiver word. The resulting output word tuple has a letter multiset that differs from the input in a single letter, and allows to recover the input uniquely, as required for injectivity.

\paragraph{Organization.}
Section~\ref{sec:preliminaries} defines the model and derives the equation for the stationary distribution. Section~\ref{sec:decomposition} presents a decomposition of excess cost as a sum of $s_j$, and Section~\ref{sec:polynomial} shows that the inequality $s_j\le p_j$ reduces to nonnegativity of a polynomial in gap variables, and derives a formula for the coefficients of the polynomial. Section~\ref{sec:injection} proves nonnegativity of each coefficient by constructing an injection. The proof of Theorem~\ref{thm:poly-convergence-inverse-poly} is deferred to Appendix~\ref{app:poly-convergence}.

\section{Model and Preliminaries}\label{sec:preliminaries}

\subsection{List Update Under IID Requests}\label{sec:listUpdatePrelim}

There are $n$ items labeled $1,\dots,n$.
Requests are i.i.d. draws from a distribution $p=(p_1,\dots,p_n)$, where $p_i\ge 0$ and $\sum_i p_i=1$.
We assume w.l.o.g.\ that $p_1 \ge p_2 \ge \dots \ge p_n$. If some $p_i=0$, those items are (almost surely) never requested and will eventually drift to the end of the list without further affecting the dynamics. We may thus assume w.l.o.g.\ that $p_i>0$ for all $i$.

A \emph{state} is a permutation $\pi\in S_n$; we interpret $\pi(i)\in\{1,\dots,n\}$ as the current position (rank) of item $i$.
The cost of serving a request to item $i$ in state $\pi$ is $\pi(i)$. The optimal ordering is the identity permutation (item $i$ at position $i$), and its expected cost is
\[
\OPT(p) \;=\; \sum_{i=1}^n ip_i.
\]

\subsection{The Transposition Rule as a Markov Chain}

\begin{definition}[Transposition]
Given a state (permutation) $\pi$ and a requested item $i$:
if $\pi(i)=1$ do nothing;
if $\pi(i)=k>1$, swap item $i$ with the item currently at position $k-1$ (the predecessor of $i$).
\end{definition}

Under i.i.d. requests, Transposition induces a time-homogeneous Markov chain on the symmetric group $S_n$ of permutations.
By the assumption $p_i>0$, the chain is irreducible and aperiodic, hence has a unique stationary distribution. Its formula was established in~\cite{Rivest76,Hendricks76} and is given by the following lemma.

\begin{lemma}[Stationary distribution~\cite{Rivest76,Hendricks76}]\label{lem:stationary}
The Transposition Markov chain is reversible with stationary distribution
\begin{equation}\label{eq:Q}
Q(\pi)\;=\;\frac{1}{\Z(p)}\prod_{i=1}^n p_i^{n-\pi(i)},
\qquad
\text{where }\;\Z(p)=\sum_{\sigma\in S_n}\prod_{i=1}^n p_i^{n-\sigma(i)}.
\end{equation}
\end{lemma}
\begin{proof}
It suffices to verify that the probability distribution $Q$ satisfies detailed balance for every allowed transition.
Let $\pi$ be a permutation and let $k\in[n-1]$.
Write $a=\pi^{-1}(k)$ and $b=\pi^{-1}(k+1)$ for the two items in adjacent positions $k$ and $k+1$.
The chain can transition from $\pi$ to the permutation $\pi'$ obtained by swapping $a$ and $b$
exactly when the requested item is $b$, which occurs with probability $p_b$.
Conversely, from $\pi'$ the chain returns to $\pi$ exactly when the requested item is $a$, which occurs with probability $p_a$. One can check directly that $Q(\pi)p_b = Q(\pi')p_a$, proving detailed balance.
\end{proof}

\begin{remark}\label{rem:gladiator}
    An analogous proof shows that the gladiator chain is reversible with the \emph{same} distribution $Q$. The only difference is that the transition probabilities from $\pi$ to $\pi'$ and vice versa are $\frac{1}{n-1}\frac{p_b}{p_a+p_b}$ and $\frac{1}{n-1}\frac{p_a}{p_a+p_b}$, respectively, and the extra divisors cancel.
\end{remark}

\section{Excess Cost Decomposition}\label{sec:decomposition}

Let $\cost(\pi)=\sum_{i=1}^n p_i\,\pi(i)$ denote the expected access cost in state $\pi$.
Then the stationary expected cost of Transposition is
\[
\E_{\pi\sim Q}[\cost(\pi)]
=\sum_{\pi\in S_n}Q(\pi)\sum_{i=1}^n p_i\,\pi(i).
\]

The next lemma rewrites the \emph{excess} over $\OPT$ as a sum over inversions.

\begin{lemma}[Inversion decomposition]\label{lem:inversions}
For any distribution $p$ and any random permutation $\pi$,
\[
\E\big[\cost(\pi)\big]-\OPT(p)
\;=\;
\sum_{1\le i<j\le n} (p_i-p_j)\Prb[\pi(j)<\pi(i)].
\]
\end{lemma}

\begin{proof}
Consider bubble-sorting a permutation $\pi$ into the identity order by swapping adjacent inverted pairs. Each adjacent swap that exchanges items $i<j$ with $\pi(j)<\pi(i)$ decreases $\cost(\pi)$ by exactly $p_i-p_j$ (because $i$ moves one step forward and $j$ one step backward).
Summing $p_i-p_j$ over all inversions yields the total decrease needed to reach the sorted order, which has cost $\OPT(p)$.
Taking expectation over $\pi$ yields the claim.
\end{proof}

Define for each $j\in\{2,\dots,n\}$ the quantity
\begin{equation}\label{eq:sj}
s_j
\;:=\;
\sum_{i=1}^{j-1}(p_i-p_j)\Prb[\pi(j)<\pi(i)]
\qquad\text{(under the stationary law $Q$).}
\end{equation}
Then Lemma~\ref{lem:inversions} implies
\begin{equation}\label{eq:excess-sum}
\E[\cost(\pi)]-\OPT(p)=\sum_{j=2}^n s_j.
\end{equation}
Intuitively, $s_j$ is the excess cost that can be blamed on item $j$ appearing too early in the permutation. 
We will show that $s_j\le p_j$ for all $j$. Theorem~\ref{thm:main} then follows immediately:
\[
\E[\cost(\pi)]-\OPT(p)=\sum_{j=2}^n s_j \le \sum_{j=2}^n p_j \le 1.
\]
The inequality $s_j\le p_j$ also directly yields Corollary~\ref{cor:gladiator}, since the expectation in Corollary~\ref{cor:gladiator} is precisely equal to $s_j$ when $\sum_i p_i=1$. The case where strengths do not sum to $1$ follows via linearity of expectation.

\section{Reduction to Coefficient Nonnegativity}\label{sec:polynomial}

Fix $j\in\{2,\dots,n\}$.
Using~\eqref{eq:Q} and~\eqref{eq:sj}, we can write
\[
s_j
=\frac{1}{\Z(p)}
\sum_{\pi\in S_n}
\left(\sum_{i=1}^{j-1}(p_i-p_j)\1_{\pi(j)<\pi(i)}\right)
\prod_{l=1}^n p_l^{n-\pi(l)}.
\]
Multiplying by $\Z(p)=\sum_{\pi\in S_n}\prod_{l=1}^n p_l^{n-\pi(l)}$, we see that $s_j\le p_j$ is equivalent to
\begin{equation}\label{eq:numerator-nonneg}
\sum_{\pi\in S_n}
\left(
p_j-\sum_{i=1}^{j-1}(p_i-p_j)\1_{\pi(j)<\pi(i)}
\right)
\prod_{l=1}^n p_l^{\,n-\pi(l)}
\;\ge\;0.
\end{equation}
Thus, our task reduces to proving nonnegativity of a polynomial.

\subsection{Gap Variables}

Introduce nonnegative ``gap'' variables
\[
x_i:=p_i-p_{i+1}\quad (1\le i\le n-1),
\qquad
x_n:=p_n.
\]
Then $p_i=\sum_{h=i}^n x_h$, and the inequalities $p_1\ge\cdots\ge p_n\ge 0$ translate to $x_i\ge 0$ for all $i$.

Under this change of variables, the left-hand side of~\eqref{eq:numerator-nonneg} can be expressed as a polynomial in $(x_1,\dots,x_n)$. The purpose of introducing the gap variables $x_1,\dots,x_n$ is to obtain a representation in which \emph{every coefficient is nonnegative}, implying the polynomial is nonnegative on $\mathbb{R}_{\ge 0}^n$.

\subsection{The Coefficient Formula}

To state the coefficient formula cleanly, define for any permutation $\pi\in S_n$ the polynomial
\[
F^\pi(x)\;:=\;\prod_{i=1}^n \left(\sum_{h=i}^n x_h\right)^{n-\pi(i)}.
\]
For a multi-index $d=(d_1,\dots,d_n)\in\mathbb{Z}_{\ge 0}^n$, let $E^\pi(d)$
denote the coefficient of the monomial $x_1^{d_1}\cdots x_n^{d_n}$ in $F^\pi(x)$.
(If some $d_i<0$, we set $E^\pi(d)=0$.)

The next lemma is the bridge from $s_j\le p_j$ to the combinatorial injection.

\begin{lemma}[Coefficient formula]\label{lem:coeff}
After switching to gap variables, the coefficient of $x_1^{d_1}\dots x_n^{d_n}$ in the polynomial corresponding to~\eqref{eq:numerator-nonneg} is
\begin{equation}\label{eq:coeff-expression}
\sum_{\pi\in S_n}
\left(
\sum_{k=j}^n E^\pi(d-e_k) \;-\;
\sum_{k=1}^{j-1} E^\pi(d-e_k)\sum_{i=1}^k \1_{\pi(i)>\pi(j)}
\right),
\end{equation}
where $e_r$ denotes the $r$th standard basis vector. In particular, the polynomial is homogeneous of degree $\binom{n}{2}+1$.
\end{lemma}
\begin{proof}
    Rewriting the polynomial in gap variables, it becomes
    \begin{align*}
        &\sum_{\pi\in S_n}
        \left(
        \sum_{k=j}^n x_k-\sum_{i=1}^{j-1}\1_{\pi(j)<\pi(i)}\sum_{k=i}^{j-1} x_k
        \right)
        \prod_{l=1}^n \left(\sum_{h=l}^{n} x_h\right)^{\,n-\pi(l)}\\
        =&\sum_{\pi\in S_n}
        \left(
        \sum_{k=j}^n x_k-\sum_{k=1}^{j-1} x_k\sum_{i=1}^k\1_{\pi(i)>\pi(j)}
        \right)
        F^\pi(x).
    \end{align*}
    The coefficient formula can be read off directly. Homogeneity of degree $\binom{n}{2}+1$ follows from the fact that $F^\pi(x)$ is homogeneous of degree $\sum_{i=1}^n (n-\pi(i))=\sum_{i=0}^{n-1} i = \binom{n}{2}$.
\end{proof}

The rest of the proof is coefficientwise: we show that for every degree vector $d$, the quantity~\eqref{eq:coeff-expression} is nonnegative. Since $x_i\ge 0$, this implies the entire polynomial is nonnegative, hence~\eqref{eq:numerator-nonneg}, hence $s_j\le p_j$.

\section{A Sign-Eliminating Injection}\label{sec:injection}

We now interpret the parts of~\eqref{eq:coeff-expression} as counts of combinatorial objects and construct an injection from the objects contributing with negative sign to those contributing with positive sign.

\subsection{Notation}
We write $\Sigma^m$ for the set of words of length $m$ over an alphabet $\Sigma$, and $\Sigma^*:=\bigcup_{m\ge 0}\Sigma^m$.  For words $u,v\in\Sigma^*$, we denote by $uv$ their concatenation. For $u\in\Sigma^*$ and $l\in\Sigma$, we denote by $\#_l(u)$ the number of occurrences of $l$ in $u$, and by $|u|$ the length of $u$. In the following, we will often use the alphabet $\Sigma_i:=\{i,i+1,\dots,n\}$ for some $i\in[n]$.

\subsection{\texorpdfstring{\boldmath Interpreting $E^\pi(\cdot)$ as Word-Tuples}{Interpreting E\string^π as Word-Tuples}}

Consider a fixed permutation $\pi\in S_n$. Expanding $\left(\sum_{l=i}^n x_l\right)^{n-\pi(i)}$ by the multinomial theorem, each monomial is of the form $\prod_{l=i}^n x_l^{\alpha_{l}}$ where $\alpha_l\in\mathbb Z_{\ge 0}$ with $\sum_{l=i}^n\alpha_l=n-\pi(i)$, and its coefficient is the multinomial coefficient 
\begin{align*}
    \binom{n-\pi(i)}{\alpha_i,\dots,\alpha_n}=\frac{(n-\pi(i))!}{\prod_{l=i}^n \alpha_l!}.
\end{align*}
This coefficient is the same as the number of words $w\in\Sigma_i^{n-\pi(i)}$ satisfying $\#_l(w)=\alpha_l$ for each $l\in\Sigma_i$. We can thus write
\begin{align*}
    \left(\sum_{l=i}^n x_l\right)^{n-\pi(i)} = \sum_{w\in\Sigma_i^{n-\pi(i)}} \prod_{l=1}^n x_l^{\#_{l}(w)}.
\end{align*}
Hence,
\begin{align*}
    F^\pi(x)&= \prod_{i=1}^n\sum_{w_i\in \Sigma_i^{n-\pi(i)}} \prod_{l=1}^n x_l^{\#_{l}(w_i)}\\
    &= \sum_{(w_1,\dots,w_n)\in\prod_{i}\Sigma_i^{n-\pi(i)}}\,\,\prod_{l=1}^n x_l^{\#_l(w_1\dots w_n)}
\end{align*}
Thus, $E^\pi(d)$ is the number of word tuples $(w_1,\dots,w_n)$ where for all $i\in[n]$,
\begin{itemize}
\item $w_i\in\Sigma_i^{n-\pi(i)}$,
\item $\#_l(w_1\dots w_n)=d_l$ for all $l\in[n]$.
\end{itemize}

\subsection{\texorpdfstring{\boldmath The Sets $A$ and $B$}{The Sets A and B}}

For the remainder of the proof, fix $d\in\mathbb{Z}_{\ge 0}^n$ with $\sum_i d_i=\binom{n}{2}+1$ and $j\in\{2,\dots,n\}$. We now define sets $A$ and $B$, such that $|A|$ is the sum of negative terms in~\eqref{eq:coeff-expression} and $|B|$ is the sum of positive terms in~\eqref{eq:coeff-expression}. As this involves terms of the form $\sum_{\pi\in S_n}E^\pi(\cdot)$, instead of considering a fixed $\pi$ we need to sum the above counts over all $\pi\in S_n$. This simply means that the previous constraint that $w_i$ has length $n-\pi(i)$ for each $i$ turns into the constraint that the lengths of $w_1,\dots,w_n$ are a permutation of $\{0,\dots,n-1\}$.

We define $B$ as the set of tuples $(w_1,\dots,w_n)\in(\Sigma_1^*)^n$ such that
{
\renewcommand{\theenumi}{(B\arabic{enumi})}
\renewcommand{\labelenumi}{\theenumi}
\begin{enumerate}[leftmargin=3.0em]
\item \label{it:alphabet} $\forall l\in[n]\colon w_l\in\Sigma_l^*$,
\item \label{it:lengths} the lengths $|w_1|,\dots,|w_n|$ are a permutation of $\{0,\dots,n-1\}$,
\item \label{it:BletterCounts} $\exists k\in\Sigma_j$ such that $\forall l\in[n]\colon\#_l(w_1\dots w_n)=d_l-\1_{l=k}$.
\end{enumerate}
}
Then
\[
|B|=\sum_{\pi\in S_n}\sum_{k=j}^n E^\pi(d-e_k).
\]
Elements of $A$ contain an integer parameter $i$ in addition to words $w_1,\dots,w_n$. Specifically, we define $A$ as the set of tuples $((w_1,\dots,w_n),i)\in(\Sigma_1^*)^n\times[j-1]$ such that
{
\renewcommand{\theenumi}{(A\arabic{enumi})}
\renewcommand{\labelenumi}{\theenumi}
\begin{enumerate}[leftmargin=3.0em]
\item \label{it:A1} $w_1,\dots,w_n$ satisfy the same alphabet and distinct-length constraints \ref{it:alphabet} and \ref{it:lengths},
\item \label{it:i} $|w_i|<|w_j|$,
\item \label{it:A3} $\exists k\in[i,j-1]$ such that $\forall l\in[n]\colon\#_l(w_1\dots w_n)=d_l-\1_{l=k}$.
\end{enumerate}
}
One checks that
\[
|A|=\sum_{\pi\in S_n}\sum_{k=1}^{j-1} E^\pi(d-e_k)\sum_{i=1}^k \1_{\pi(i)>\pi(j)},
\]
where the extra parameter $i$ in the tuple along with the conditions $|w_i|<|w_j|$ and $k\ge i$ accounts for the sum over $i\in [k]$ with $\pi(i)>\pi(j)$.

Thus, proving nonnegativity of \eqref{eq:coeff-expression} reduces to proving $|A|\le |B|$. We show this by constructing an explicit injection $f\colon A\hookrightarrow B$.

\subsection{Construction of the Injection}

We describe an injection $f\colon A\hookrightarrow B$ by specifying an algorithm that transforms an input from $A$ into an output from $B$. We motivate and explain the design of the map $f$ in this section, provide a summary of its formal definition in Algorithm~\ref{alg:injection}, and an example in Example~\ref{ex:injection} and Figure~\ref{fig:example}. Sections~\ref{sec:welldefined} and \ref{sec:injective} prove that the map is indeed well-defined and injective.

By definition, an input from $A$ is of the form $((w_1,\dots,w_n),i)$, and has a unique associated ``deficit'' letter $k\in[i,j-1]$, i.e., satisfying $\#_l(w_1\cdots w_n)=d_l-\mathbf{1}_{l=k}$.
Write
\[
L:=|w_i|,
\qquad
U:=|w_j|,
\qquad
m:=U-L,
\]
and note that $m\ge 1$ by \ref{it:i}.
The map must output $(\tw_1,\dots,\tw_n)$ whose lengths are still a permutation of $\{0,\dots,n-1\}$ and whose total letter counts
have a unique deficit in $\Sigma_j=[n]\setminus[j-1]$ (instead of the input deficit $k\in[j-1]$).

The construction only modifies words whose lengths lie in the interval $[L,U]$. All other words are copied unchanged (lines~\ref{ln:forDefault}--\ref{ln:twlwl} of Algorithm~\ref{alg:injection}).

Injectivity requires recovering the extra parameter $i$ from the output. We will achieve this indirectly by encoding the length $L=|w_i|$ as the new length of the word $\tw_j$ in the output. To do this, we truncate $w_j$ to length $L$ by writing
\[
w_j=\tw_jb_1b_2\cdots b_m
\qquad
(\tw_j\in\Sigma_j^{L},\ b_t\in\Sigma_j).
\]
This makes $L$ recoverable simply as $|\tw_j|$. Consequently, if we can recover $(w_1,\dots,w_n)$, then we can also recover $i$ as the index of the word among them of length $L$.

We keep the removed letters $b_1,\dots,b_m$ as ``tokens'' that must (mostly) be reinserted elsewhere to preserve the multiset of letters. The guiding objective is to \emph{repair} the deficit at $k$ by inserting one copy of $k$ somewhere, while deliberately leaving \emph{one} of the token letters missing in the output. This shifts the unique deficit from a letter in $[j-1]$ (in $A$) to a letter in $\Sigma_j$ (as required in $B$).

It remains to specify how to modify the words of input lengths in $[L,U-1]$. Let $i_1,\dots,i_m$ be the indices of these words sorted by word length, i.e., such that $|w_{i_t}|=L+t-1$. The new output lengths of these words must cover the interval $[L+1,U]$ in order to restore the property that the multiset of all lengths equals $\{0,1,\dots,n-1\}$. We partition the indices $i_1,\dots,i_m$ into three groups:
\begin{itemize}
    \item \emph{Receiver} indices: $i_t\le k$
    \item \emph{Exchange} indices: $k<i_t<j$
    \item \emph{Tail} indices: $i_t>j$
\end{itemize}
This yields a correct partition since none of the indices $i_t$ equals $j$, as $|w_j|=U\notin[L,U-1]$.

\begin{itemize}
    \item The tail words (with index $i_t>j$) remain unchanged: $\tw_{i_t}:=w_{i_t}$.
    
    \item Each exchange word ($k<i_t<j$) also remains unchanged \emph{in length}, but we will replace its final letter $a_t$ by the token $b_t$. The alphabet constraints permit this letter exchange since $b_t\in\Sigma_j\subset \Sigma_{i_t}$. The purpose (as we will leverage below) is to ensure that both tail and exchange words end with a letter from $\Sigma_j$ in the output.

    \item Receiver words ($i_t\le k$) are the only words that grow in length. Concretely, every receiver word grows to the length of the next larger input receiver word, and the longest receiver word grows to length $U$. These extensions are achieved by inserting the original deficit letter $k$, letters $a_{t'}$ that were swapped out of exchange words, and other token letters $b_{t'}$.
\end{itemize}

The exact update process is specified in lines~\ref{ln:vgetsk}--\ref{ln:endmloop} of Algorithm~\ref{alg:injection}. The algorithm processes the indices from $t=m$ down to $1$ while maintaining a buffer string variable $v$, initialized to $v\gets k$. Intuitively, at any moment $v$ contains letters that have been ``dislodged'' from some already-processed word (or taken from the removed suffix of $w_j$) and are waiting to be inserted into a receiver word. A tail step $t$ prepends a token $b_t$ from the removed suffix to $v$. An exchange step $t$ prepends the swapped-out letter $a_t$ to $v$. A receiver step $t$ appends the entire buffer $v$ to the end of $w_{i_t}$, and resets the buffer to the single letter $b_t$.

Because $i_1=i\le k$, the final iteration $t=1$ is always a receiver step, so after the loop terminates we have $v=b_1$,
and this last letter is never inserted anywhere.  Thus the deficit letter moves from $k\in[j-1]$ in the input to $b_1\in\Sigma_j$ in the output.

The initial deficit letter $k$ starts inside $v$ and is inserted the first time a receiver step occurs, which creates the output word of length $U$. Every other word $\tw_{i_t}$ ends with a letter from $\Sigma_j$. Thus, we can identify $U$ from the output as the smallest length $>L$ whose word ends in a letter from $[j-1]$, and this letter must be $k$. This further allows to recover $m=U-L$, the index set $\{i_1,\dots,i_m\}$ as well as its partition into tail, exchange and receiver indices. As a result, it is possible to simulate the algorithm in reverse and uniquely determine the input $((w_1,\dots,w_n),i)$ from the output $(\tw_1,\dots,\tw_n)$.

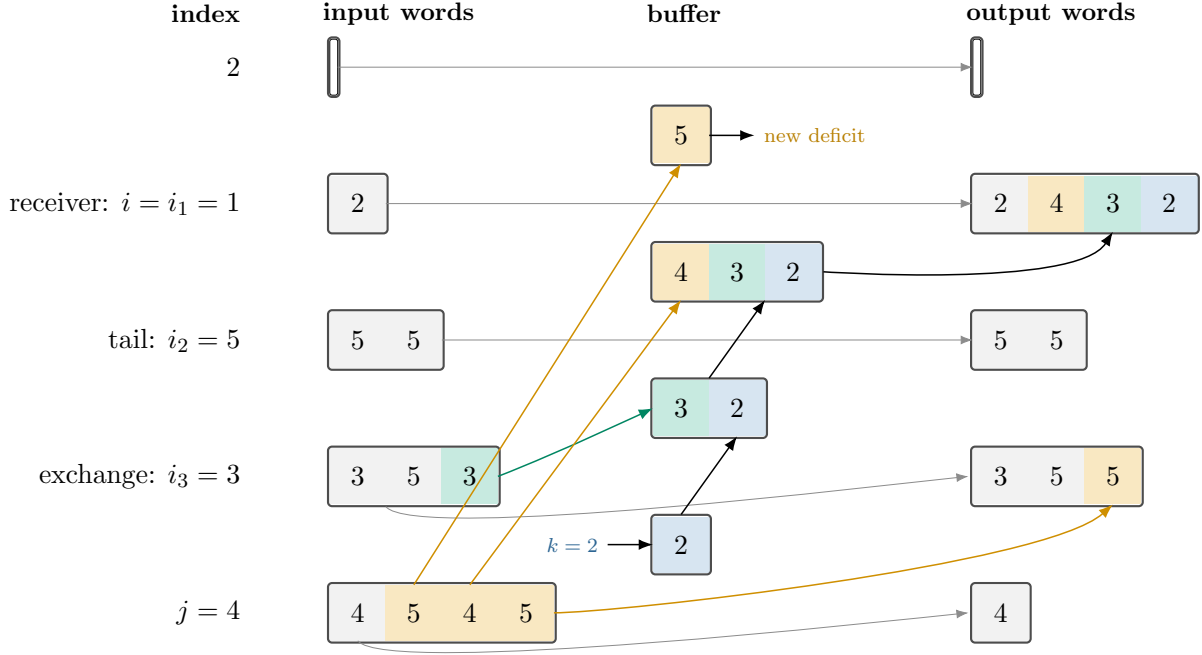
\begin{figure}[htb]
\resizebox{\columnwidth}{!}{

\begin{tikzpicture}[
  x=1cm,y=1cm,
  show background rectangle,
  background rectangle/.style={fill=white},
  slot/.style={minimum width=7.8mm, minimum height=7.8mm, inner sep=0pt, outer sep=0pt, font=\normalsize, draw=none},
  plain/.style={slot, fill=black!5},
  token/.style={slot, fill=tokorange!24},
  intro/.style={slot, fill=introblue!20},
  exchange/.style={slot, fill=exchgreen!22},
  emptyword/.style={draw, rounded corners=1pt, minimum width=1.1mm, minimum height=7.8mm, inner sep=0pt, outer sep=0pt},
  header/.style={font=\bfseries\small, align=center},
  rowlab/.style={font=\small, align=left, text width=0.3cm},
  smallnote/.style={font=\scriptsize, align=center},
  wordbox/.style={rounded corners=1.6pt, draw=black!70, line width=.85pt, inner sep=.7pt},
  tokenbrace/.style={decorate,decoration={brace,amplitude=4pt}},
  move/.style={-Latex, semithick},
  stay/.style={-Latex, thin, draw=black!45},
  emphasis/.style={draw=red!70!black, rounded corners=2pt, dashed, line width=.8pt, inner sep=1.2pt}
]

% column positions
\def\xrole{0.1}
\def\xidx{2.85}
\def\xin{4.35}
\def\xbuf{8.85}
\def\xout{13.3}
\def\dx{7.8mm}

% row positions (sorted by input length)
\def\yTwo{0.0}
\def\yOne{-1.9}
\def\yFive{-3.8}
\def\yThree{-5.7}
\def\yFour{-7.6}

% buffer states
\def\yBufStart{-6.65}
\def\yBufTwo{-4.75}
\def\yBufThree{-2.85}
\def\yBufFinal{-0.95}

% title and headers
\node[header,anchor=east] at (\xidx,0.75) {index};
\node[header] at (4.92,0.75) {input words};
\node[header] at (8.89,0.75) {buffer};
\node[header] at (14.00,0.75) {output words};

% index labels
\foreach \yy/\lab in {\yTwo/$2$,\yOne/receiver: $i=i_1=1$,\yFive/tail: $i_2=5$,\yThree/exchange: $i_3=3$,\yFour/$j=4$}{
  \node[anchor=east] at (\xidx,\yy) {\lab};
}

% input words
\node[emptyword] (in2e) at ($(\xin,\yTwo)+(-3.35mm,0)$) {};
\node[plain] (in1a) at (\xin,\yOne) {$2$};
\node[plain] (in5a) at (\xin,\yFive) {$5$};
\node[plain] (in5b) at ($(in5a)+(\dx,0)$) {$5$};
\node[plain] (in3a) at (\xin,\yThree) {$3$};
\node[plain] (in3b) at ($(in3a)+(\dx,0)$) {$5$};
\node[exchange] (in3c) at ($(in3b)+(\dx,0)$) {$3$};
\node[plain] (in4a) at (\xin,\yFour) {$4$};
\node[token] (in4b) at ($(in4a)+(\dx,0)$) {$5$};
\node[token] (in4c) at ($(in4b)+(\dx,0)$) {$4$};
\node[token] (in4d) at ($(in4c)+(\dx,0)$) {$5$};

% output words (same row order as the index column)
\node[emptyword] (out2e) at ($(\xout,\yTwo)+(-3.35mm,0)$) {};
\node[plain] (out1a) at (\xout,\yOne) {$2$};
\node[token] (out1b) at ($(out1a)+(\dx,0)$) {$4$};
\node[exchange] (out1c) at ($(out1b)+(\dx,0)$) {$3$};
\node[intro] (out1d) at ($(out1c)+(\dx,0)$) {$2$};
\node[plain] (out5a) at (\xout,\yFive) {$5$};
\node[plain] (out5b) at ($(out5a)+(\dx,0)$) {$5$};
\node[plain] (out3a) at (\xout,\yThree) {$3$};
\node[plain] (out3b) at ($(out3a)+(\dx,0)$) {$5$};
\node[token] (out3c) at ($(out3b)+(\dx,0)$) {$5$};
\node[plain] (out4a) at (\xout,\yFour) {$4$};

% word outlines
\node[wordbox, fit=(in2e)] {};
\node[wordbox, fit=(in1a)] {};
\node[wordbox, fit=(in5a)(in5b)] {};
\node[wordbox, fit=(in3a)(in3c)] {};
\node[wordbox, fit=(in4a)(in4d)] {};
\node[wordbox, fit=(out2e)] {};
\node[wordbox, fit=(out1a)(out1d)] {};
\node[wordbox, fit=(out5a)(out5b)] {};
\node[wordbox, fit=(out3a)(out3c)] {};
\node[wordbox, fit=(out4a)] {};

% unchanged arrows
\draw[stay] (in2e.east) -- (out2e.west);
\draw[stay] ($(in1a.east)+(0,0)$) -- ($(out1a.west)+(-0,0)$);
\draw[stay] (in5b.east) -- (out5a.west);
\draw[stay] ($(in3a.south)!0.5!(in3b.south)+(0,-0.03)$) .. controls +(0.4,-0.5) and +(-1,-0.1) .. ($(out3a.west)+(-0.06,0)$);
\draw[stay] (in4a.south) .. controls +(0.4,-0.5) and +(-1,-0.1) .. ($(out4a.west)+(-0.06,0)$);

% buffer states
\node[intro] (buf0) at (\xbuf,\yBufStart) {$2$};
\node[smallnote, left=18pt of buf0, text=introblue!80!black, align=right] (klabel) {$k=2$};
\node[exchange] (buf1a) at (\xbuf,\yBufTwo) {$3$};
\node[intro]    (buf1b) at ($(buf1a)+(\dx,0)$) {$2$};
\node[token]    (buf2a) at (\xbuf,\yBufThree) {$4$};
\node[exchange] (buf2b) at ($(buf2a)+(\dx,0)$) {$3$};
\node[intro]    (buf2c) at ($(buf2b)+(\dx,0)$) {$2$};
\node[token] (buf3) at (\xbuf,\yBufFinal) {$5$};
\node[smallnote, right=18pt of buf3, text=tokorange!80!black, align=left] (newdeficitlabel) {new deficit};
\node[wordbox, fit=(buf0)] {};
\node[wordbox, fit=(buf1a)(buf1b)] {};
\node[wordbox, fit=(buf2a)(buf2c)] {};
\node[wordbox, fit=(buf3)] {};

% buffer progression arrows
\draw[move] (klabel.east) -- (buf0.west);
\draw[move] ($(buf0.north)+(0,0.03)$) -- (buf1b.south);
\draw[move] ($(buf1a.north)!0.5!(buf1b.north)+(0,0.03)$) -- ($(buf2b.south)!0.5!(buf2c.south)$);
\draw[move] (buf3.east) -- (newdeficitlabel.west);

% arrows from input to moved pieces / buffer
\draw[move, draw=exchgreen!85!black] (in3c.east) .. controls +(0.9,0.35) and +(0,0) .. (buf1a.west);
\draw[move, draw=tokorange!90!black] (in4d.east) .. controls +(0.85,0) and +(-0.95,-0.85) .. (out3c.south);
\draw[move, draw=tokorange!90!black] (in4c.north) .. controls +(0.6,0.8) and +(-0.6,-1.1) .. (buf2a.south);
\draw[move, draw=tokorange!90!black] (in4b.north) .. controls +(0.6,0.8) and +(-0.6,-1.2) .. (buf3.south);

% buffer appended to receiver word
\draw[move] ($(buf2c.east)+(0.03,0)$) .. controls +(0,0) and +(-0.7,-0.7) .. ($(out1b.south)!0.5!(out1d.south)$);
\end{tikzpicture}
}
\caption{Illustration of Example~\ref{ex:injection}. Rows are sorted by input word length.}\label{fig:example}
\end{figure}

\begin{example}\label{ex:injection}
    See Figure~\ref{fig:example} for an illustration. Consider $n=5$ with $d=(0,2,2,2,5)$, $j=4$, and input $((w_1,w_2,w_3,w_4,w_5),i)\in A$ with
    \begin{align*}
        w_1&=2\\
        w_2&=\epsilon\\
        w_3&=353\\
        w_4&=4545\\
        w_5&=55\\
        i&=1,
    \end{align*}
    where $\epsilon$ denotes the empty word. Properties \ref{it:A1}--\ref{it:A3} are satisfied with deficit letter $k=2$. We have $L=1$, $U=4$ and $m=3$. Since $w_2$ is the only word whose length does not fall in $[L,U]$, the algorithm sets
    \[\tw_2\gets w_2=\epsilon\] in lines~\ref{ln:forDefault}--\ref{ln:twlwl}. In line~\ref{ln:j} it sets \[\tw_4\gets 4\quad\text{ and }\quad(b_1,b_2,b_3)\gets(5,4,5).\]
    The buffer is initialized as $v\gets k=2$. The remaining words are processed in decreasing order of length:
    \begin{itemize}
        \item $i_3=3$ (exchange step) sets $\tw_3\gets 355$ and updates $v\gets32$.
        \item $i_2=5$ (tail step) sets $\tw_5\gets 55$ and updates $v\gets432$.
        \item $i_1=1$ (receiver step) sets $\tw_1\gets 2432$ and resets $v\gets 5$.
    \end{itemize}
    The output $(\tw_1,\tw_2,\tw_3,\tw_4,\tw_5)=(2432,\epsilon,355,4,55)$ satisfies \ref{it:alphabet}--\ref{it:BletterCounts} with new deficit letter $\tilde k=b_1=5$.
\end{example}

\begin{algorithm}[htb]
\caption{Sign-eliminating injection $f\colon A\hookrightarrow B$}
\label{alg:injection}
\begin{algorithmic}[1]
\Fixed $d\in\mathbb{Z}_{\ge 0}^n$ with $\sum_i d_i=\binom{n}{2}+1$ and $j\in\{2,\dots,n\}$
\Require $((w_1,\dots,w_n),i)\in A$
\Ensure $(\tw_1,\dots,\tw_n)\in B$
    \State Let $k\in[j-1]$ such that $\forall l\in[n]\colon\#_l(w_1\dots w_n)=d_l-\1_{l=k}$
    \State $L\gets |w_i|$ and $U\gets |w_j|$
    \State $m\gets U-L$\label{ln:m}
    \For{$l\in[n]$ with $|w_l|\notin[L,U]$}\label{ln:forDefault}
        \State $\tw_l\gets w_l$\label{ln:twlwl}
    \EndFor
    \State Let $\tw_j\in\Sigma_j^{L}$ and $b_1,\dots,b_m\in\Sigma_j$ such that $w_j=\tw_jb_1\dots b_m$\label{ln:j}
    \State Let $i_1,i_2,\dots,i_m$ such that $|w_{i_{t}}|=L+t-1$
    \State $v\gets k$\label{ln:vgetsk}
    \For{$t=m$ \textbf{down to} $1$}
        \If{$i_t>j$} \Comment{$i_t>j:$ Tail}
            \State $\tw_{i_t}\gets w_{i_t}$
            \State $v\gets b_t v$
        \ElsIf{$i_t>k$} \Comment{$k<i_t<j$: Exchange}
            \State Let $u\in\Sigma_{i_t}^*$ and $a_t\in\Sigma_{i_t}$ such that $w_{i_t}=ua_t$
            \State $\tw_{i_t}\gets u b_t$
            \State $v\gets a_t v$
        \Else\Comment{$i_t\le k$: Receiver}
            \State $\tw_{i_t}\gets w_{i_t}v$
            \State $v\gets b_t$
        \EndIf
    \EndFor\label{ln:endmloop}
    \State \Return $(\tw_1,\dots,\tw_n)$
\end{algorithmic}
\end{algorithm}

The remaining sections prove that the map is well-defined and injective.

\subsection{Well-Definedness}\label{sec:welldefined}

\begin{lemma}[Well-definedness of $f$]\label{lem:f-welldefined}
For every input $((w_1,\dots,w_n),i)\in A$, Algorithm~\ref{alg:injection} returns an output
$(\tw_1,\dots,\tw_n)\in B$.
\end{lemma}

\begin{proof}
Fix an input $((w_1,\dots,w_n),i)\in A$ and let $(\tw_1,\dots,\tw_n)$ be the output.

\paragraph{\ref{it:alphabet} Alphabet constraints.}
Every word $\tw_l$ must lie in $\Sigma_l^*$.
Words with $|w_l|\notin[L,U]$ are unchanged, so the constraint holds for them.
Also, the word $\tw_j$ is a prefix of $w_j$, hence $\tw_j\in\Sigma_j^*$.

Now consider $l=i_t$.
If $i_t>j$, then $\tw_{i_t}=w_{i_t}\in\Sigma_{i_t}^*$.
If $k<i_t<j$, we set $\tw_{i_t}=u b_t$ with $u\in\Sigma_{i_t}^*$ and $b_t\in\Sigma_j$. Since $j> i_t$, also $b_t\in\Sigma_{i_t}$.
If $i_t\le k$, then $\tw_{i_t}=w_{i_t}v$. Every letter ever placed into $v$ is $\ge k\ge i_t$. Thus, $v\in\Sigma_{i_t}^*$ and $\tw_{i_t}\in\Sigma_{i_t}^*$.

\paragraph{\ref{it:lengths} Word lengths.}
All lengths outside the interval $[L,U]$ are unchanged.
It suffices to show that the multiset of output lengths among indices whose \emph{input} lengths lie in $[L,U]$
is exactly $\{L,L{+}1,\dots,U\}$.

Let $R:=\{t\in[m]: i_t\le k\}$ be the set of receiver time steps.
Since $i_1=i\le k$, we have $1\in R$.
Write the receiver times in decreasing order as
\[
r_1>r_2>\cdots>r_s=1,\qquad\text{and set}\qquad r_0:=m+1.
\]
Between two consecutive receiver times, the buffer length increases by exactly $1$ at each step
(because both an exchange step and a tail step prepend a single letter to $v$),
and at each receiver time it is reset to length $1$.
Therefore, at time $t=r_q$ the buffer has length $|v|=r_{q-1}-r_q$.
Hence, the receiver word $w_{i_{r_q}}$ (whose input length is $L+r_q-1$) grows to output length
\[
|\tw_{i_{r_q}}|
=(L+r_q-1) + (r_{q-1}-r_q)
= L+r_{q-1}-1.
\]
All non-receiver indices $i_t$ keep their length $L+t-1$, and $\tw_j$ has length $L$.
Thus the output lengths among $\{j,i_1,\dots,i_m\}$ are
\[
\{L\}\ \cup\ \{L+t-1:\ t\in[m]\setminus R\}\ \cup\ \{L+r_{q-1}-1:\ q=1,\dots,s\}
=\{L,L+1,\dots,U\}.
\]
Therefore, the multiset of lengths remains a permutation of $\{0,1,\dots,n-1\}$.

\paragraph{\ref{it:BletterCounts} Letter counts.}
In the input, the count of letter $l$ equals $d_l-\mathbf{1}_{l=k}$.
Truncating $w_j$ removes the letters $b_1,\dots,b_m$, which are later inserted in either the buffer or another word. In the loop,
\begin{itemize}
\item an exchange step replaces one letter $a_t$ by $b_t$ in a word and prepends $a_t$ to the buffer;
\item a tail step prepends $b_t$ to the buffer, without changing any word;
\item a receiver step appends the entire current buffer $v$ to a word, and then resets $v$ to the single letter $b_t$.
\end{itemize}
Thus, apart from the initial truncation of $w_j$, letters are never deleted: they are either moved between words and the buffer,
or inserted from the removed suffix letters $b_t$ by placing them into words (exchange) or the buffer (tail/receiver).

Because $t=1$ is a receiver step, after the last iteration we have $v=b_1$ and the algorithm terminates,
so $b_1$ is never appended to any word.  Every other removed letter $b_t$ with $t\ge 2$ is inserted exactly once:
either immediately as the new last letter of an exchange word, or by being carried in the buffer and appended at the next receiver step.
Moreover, the initial deficit letter $k$ is inserted exactly once: it starts in $v$ and is appended at the first receiver step encountered
when scanning $t=m$ down to $1$.

Therefore, the deficit at $k$ is repaired and the unique deficit letter in the output is $b_1\in\Sigma_j$:
\[
\#_l(\tw_1\cdots\tw_n)=d_l-\mathbf{1}_{l=b_1}.
\]
Hence $(\tw_1,\dots,\tw_n)\in B$.
\end{proof}

\subsection{Proof of Injectivity}\label{sec:injective}

\begin{lemma}[Injectivity of $f$]\label{lem:f-injective}
The map $f\colon A\to B$ defined by Algorithm~\ref{alg:injection} is injective.
\end{lemma}

\begin{proof}
We describe a procedure to uniquely recover the input from any output of $f$.
Fix an output $(\tw_1,\dots,\tw_n)\in f(A)$.

\paragraph{\boldmath Step 1: recover $b_1$, $L$, $U$, $m$ and $k$.}
Since $\sum_l d_l=\binom{n}{2}+1$, while $\sum_l \#_l(\tw_1\cdots\tw_n)=\binom{n}{2}$,
there is a unique letter $b_1$ such that $\#_{b_1}(\tw_1\cdots\tw_n)=d_{b_1}-1$.
This is exactly the deficit letter of the output, and by the reasoning in the preceding proof, it must equal the letter $b_1$ removed from $w_j$.

Let $L:=|\tw_j|$.
Define $U$ to be the \emph{smallest} integer $>L$ such that the (unique) word of length $U$ ends with a letter in $[j-1]$,
and let $k$ be that last letter.
For outputs in the image of $f$, such a $U$ exists and is unique:
indeed, the first receiver step in the forward loop produces a word of length $U=|w_j|$ ending with the inserted letter $k\in[j-1]$,
while for every length $l$ with $L<l<U$ the unique word of length $l$ ends with a letter in $\Sigma_j$
(exchange words end with $b_t\ge j$, tail words end with letters $>j$, and later receiver words end with $b_t\ge j$).
Thus the above scanning rule recovers the original $U$ and $k$.
Set $m:=U-L$.

\paragraph{\boldmath Step 2: recover $i$ and $i_1,\dots,i_m$.}
Let
\[
I:=\{l\in[n]:\ L\le |\tw_l|\le U\}.
\]
For outputs in $f(A)$ we have $I=\{j,i_1,\dots,i_m\}$.

Define the receiver index set
\[
R:=\{l\in I:\ l\le k\}.
\]
This is correct because in the forward algorithm the receiver condition is exactly $i_t\le k$,
while the other two cases satisfy $i_t>k$.

Order the receiver indices by increasing output length:
\[
R=\{r_1,r_2,\dots,r_s\}
\quad\text{such that}\quad
|\tw_{r_1}|<|\tw_{r_2}|<\dots<|\tw_{r_s}|.
\]
In the forward execution, $r_1=i_1=i$ (the receiver in the final step $t=1$) and $r_s$ is the first receiver encountered when scanning downward,
whose output length equals $U$.  In particular, the input index is recovered as
\[
i:=r_1.
\]

Next, define an \emph{original length} function $l_{\mathrm{orig}}:I\to\{L,\dots,U\}$ by
\[
l_{\mathrm{orig}}(r_1):=L,\qquad
l_{\mathrm{orig}}(r_{q+1}):=|\tw_{r_q}|\,\,\text{ for }q=1,\dots,s-1,\qquad
l_{\mathrm{orig}}(j):=U,
\]
and for $l\in I\setminus(R\cup\{j\})$ (i.e., exchange or tail indices) set $l_{\mathrm{orig}}(l):=|\tw_l|$.

Then for each $t\in[m]$, we recover $i_t$ as the unique element of $I\setminus\{j\}$ satisfying
\[
l_{\mathrm{orig}}(i_t)=L+t-1.
\]

\paragraph{\boldmath Step 3: reverse the loop and recover $w_{i_t}$ and $b_t$.}
We now reconstruct the words $w_{i_1},\dots,w_{i_m}$ and the removed letters $b_1,\dots,b_m$ by reversing the loop of Algorithm~\ref{alg:injection}.

Initialize a buffer string $v\gets b_1$; this equals the forward buffer after finishing the last iteration $t=1$.
For $t=1,2,\dots,m$ do:
\begin{itemize}
\item If $i_t\le k$ (receiver), then the forward step set $v\gets b_t$ and produced
$\tw_{i_t}=w_{i_t}v_{\mathrm{before}}$, where $|w_{i_t}|=L+t-1$ and $v_{\mathrm{before}}$ was the previous buffer.
In reverse we set $b_t\gets v$ (a single letter), split $\tw_{i_t}$ into prefix and suffix as
\[
\ \tw_{i_t} = w_{i_t}v_{\mathrm{before}},
\qquad |w_{i_t}|=L+t-1,
\]
and then update $v\gets v_{\mathrm{before}}$.
\item If $i_t>j$ (tail), then the forward step left the word unchanged and updated $v\gets b_t v_{\mathrm{before}}$.
In reverse we set $w_{i_t}\gets \tw_{i_t}$, let $b_t$ be the first letter of the current buffer $v$, and delete it from $v$.
\item Otherwise ($k<i_t<j$, exchange), the forward step wrote $w_{i_t}=u a_t$ and $\tw_{i_t}=u b_t$ and updated
$v\gets a_t v_{\mathrm{before}}$.  In reverse we set $b_t$ to be the last letter of $\tw_{i_t}$, let $u$ be the remaining prefix,
let $a_t$ be the first letter of $v$, delete it from $v$, and then set $w_{i_t}\gets u a_t$.
\end{itemize}
After finishing $t=m$, the buffer must equal $v=k$ (matching the forward initialization).

\paragraph{\boldmath Step 4: recover $w_j$ and all remaining words.}
Set
\[
w_j \gets \tw_jb_1b_2\dots b_m.
\]
For any index $l$ with $|\tw_l|\notin[L,U]$, the forward algorithm left the word unchanged, so we set $w_l\gets\tw_l$.
This reconstructs a unique candidate preimage $((w_1,\dots,w_n),i)$.

\paragraph{Step 5: conclude injectivity.}
By construction, running Algorithm~\ref{alg:injection} on the reconstructed $((w_1,\dots,w_n),i)$ returns exactly
$(\tw_1,\dots,\tw_n)$, since each reverse step undoes the corresponding forward step. All parameters $L,U,m,k,i_1,\dots,i_m,b_1,\dots,b_m$ and the input $((w_1,\dots,w_n),i)\in A$
have been reconstructed uniquely.
Therefore, every element of $f(A)$ has a unique preimage, and $f$ is injective.
\end{proof}

\section{Conclusion and Future Directions}\label{sec:conclusion}

We proved that Transposition achieves stationary expected cost at most $\OPT(p)+1$ for every i.i.d. access distribution. Thus, the simplest swap-only self-organizing heuristic essentially matches the performance of knowing $p$. Our result subsumes the multiplicative $1+o(1)$ guarantee previously known for the Zipfian case and asymptotically confirms Rivest's conjecture. Beyond list update, it implies a memoryless rule for approximately sorting probabilities given sampling access, and a quantitative guarantee on the ranking achieved by the gladiator chain.

When access probabilities are at least $n^{-O(1)}$, the cost bound holds already after poly$(n)$ steps. Our proof of this fact (Appendix~\ref{app:poly-convergence}) is based on showing polynomial mixing time of the underlying Markov chain. Note that polynomial mixing fails without a lower bound on the $p_i$, since the relative order of two items cannot change until at least one of them is requested. Nonetheless, a prior version of this article formulated the conjecture that an $\OPT+O(1)$ cost bound holds after poly$(n)$ steps in the general case, even without any lower bound on access probabilities. This conjecture was recently confirmed by Mizgerd~\cite{Mizgerd26}, who presented a proof that the cost is at most $\OPT+1250$ after $n^{29}$ steps. We further conjecture that the additive $1250$ term can be replaced by $1$, and already after a much smaller polynomial number of steps.

Another intriguing question is whether similar additive bounds are achievable for other data structures, such as self-organizing binary search trees (BSTs). A result of this type was obtained in \cite{GolinILMN18}, but the algorithm uses significant extra memory.

\paragraph{Acknowledgments.} 
I thank Sandeep Silwal for discussions about Rivest’s conjecture at the Dagstuhl Seminar 25471, ``Online Algorithms beyond Competitive Analysis'', David Geldbach for pointing out a relevant reference, and Jim Fill for many valuable comments.

This research was assisted by the large language model ChatGPT Pro. During the brainstorming phase, it suggested that the inequality $s_j\le p_j$ might hold and hypothesized, based on experiments with small $n$, that coefficients of the corresponding polynomial in gap variables appear to be nonnegative. It was, however, unable to prove either of these statements, and made no meaningful progress towards constructing the combinatorial injection. Nonetheless, these suggestions were essential for motivating the proof approach that ultimately turned out to be successful. ChatGPT Pro was also used to assist with drafting portions of the manuscript throughout the paper, including TikZ code for Figure~\ref{fig:example}, all of which have been revised by the author. The author assumes responsibility for all content.

Funded by the European Union (ERC, CCOO, 101165139). Views and opinions expressed are however those of the author(s) only and do not necessarily reflect those of the European Union or the European Research Council. Neither the European Union nor the granting authority can be held responsible for them.

% \newpage
\bibliographystyle{alpha}
\bibliography{references}

\appendix

\section{Polynomial Convergence Under Polynomial Access Probabilities}
\label{app:poly-convergence}

We show that, if access probabilities are bounded below by $n^{-O(1)}$, then the cost bound $\OPT+1$ holds already after polynomially many requests, proving Theorem~\ref{thm:poly-convergence-inverse-poly}. The proof is based on a recent lower bound on the spectral gap of regular biased adjacent-transposition chains~\cite{GreavesZ26} (which include gladiator chains) and the fact that, if access probabilities are polynomially bounded, then transition probabilities of the Transposition chain differ only polynomially from transition probabilities of the corresponding gladiator chain.

Let \(P\) denote the transition matrix of the Transposition Markov chain.
Let \(\pi_t\in S_n\) be the state after \(t\) requests.

We first recall some standard Markov-chain terminology.  For two
probability distributions \(\mu,\nu\) on a finite set \(\Omega\), their total
variation distance is
\[
    \|\mu-\nu\|_{\mathrm{TV}}
    :=
    \frac12\sum_{x\in\Omega}|\mu(x)-\nu(x)|.
\]
Let \(M\) be the transition matrix of a reversible, irreducible Markov chain on $\Omega$ with stationary distribution $Q$. Then the eigenvalues of $M$ fall into the real interval $[-1,1]$ and include $1$ \cite[Lemma 12.1 \& Lemma 12.2]{LevinP17}. If we label them in decreasing order by $1 = \theta_1\ge \theta_2\ge\dots$, then the \emph{spectral gap} of $M$ is $\lambda(M):=1-\theta_2$, and the \emph{absolute spectral gap} is
\[
    \gamma(M):=
    1-\max\{|\theta|:\theta\text{ is an eigenvalue of }M,\ \theta\ne 1\}.
\]
By \cite[Lemma 13.7 \& Remark 13.8]{LevinP17}, the spectral gap can also be written as
\[
    \lambda(M)
    =
    \min_{\substack{f\colon\Omega\to\mathbb R\\ \mathrm{Var}_Q(f)>0}}
    \frac{\mathcal{E}_M(f)}{\mathrm{Var}_Q(f)},
\]
where
\[
    \mathcal{E}_M(f)
    :=
    \frac12\sum_{x,y\in\Omega}
        Q(x)M(x,y)\bigl(f(x)-f(y)\bigr)^2
\]
denotes the expected
squared change of a function \(f\) in one step of the chain, up to the
factor \(1/2\), and
\[
    \mathrm{Var}_Q(f)
    :=
    \sum_{x\in\Omega}Q(x)
    \left(f(x)-\sum_{y\in\Omega}Q(y)f(y)\right)^2 .
\]

A standard estimate (cf. \cite[Proposition 3]{DiaconisS91}) says that, for every initial state \(x\),
\begin{align}
    \|M^t(x,\cdot)-Q\|_{\mathrm{TV}}
    \le
    \frac12 Q(x)^{-1/2}(1-\gamma(M))^t ,
    \label{eq:spectral-tv-bound}
\end{align}
where $M^t(x,\cdot)$ denotes the state distribution after $t$ steps of the Markov chain.

\begin{proof}[Proof of Theorem~\ref{thm:poly-convergence-inverse-poly}]
The cases \(n\le 2\) are immediate, so assume \(n\ge 3\).  Let
\[
    Q(\pi)
    =
    \frac1{Z(p)}
    \prod_{i=1}^n p_i^{\,n-\pi(i)}
\]
be the stationary distribution of Transposition.

For \(k\in[n-1]\), let \(\pi^{(k)}\) be the permutation obtained from
\(\pi\) by swapping the two items in positions \(k\) and \(k+1\).  Put
\[
    a=\pi^{-1}(k),
    \qquad
    b=\pi^{-1}(k+1).
\]
Thus \(a\) is immediately before \(b\) in the list.  In the Transposition
chain, the transition from \(\pi\) to \(\pi^{(k)}\) occurs exactly when
item \(b\) is requested, and hence
\begin{align}
    P(\pi,\pi^{(k)})=p_b .
    \label{eq:P-transition}
\end{align}

Recall that the gladiator chain \(G\) associated to strengths $p_i$ chooses at each step a uniformly random index  \(k\in[n-1]\) and then compares the two adjacent
items \(a=\pi^{-1}(k)\) and \(b=\pi^{-1}(k+1)\).  The item \(b\) moves ahead
of \(a\) with probability \(p_b/(p_a+p_b)\).  Equivalently,
\begin{align}
    G(\pi,\pi^{(k)})
    =
    \frac1{n-1}\cdot\frac{p_b}{p_a+p_b}.
    \label{eq:G-transition}
\end{align}
This is the special case of a so-called biased adjacent-transposition chain with comparison
probabilities
\[
    q_{ij}:=\frac{p_i}{p_i+p_j}.
\]
Very recently, Greaves and Zhu~\cite{GreavesZ26} proved a lower bound on the spectral gap of such chains provided $q$ satisfies the following regularity conditions: for \(i<j\),
\(q_{ij}\ge 1/2\), \(q_{i-1,j}\ge q_{ij}\), and
\(q_{i,j+1}\ge q_{ij}\). Since \(p_1\ge p_2\ge\cdots\ge p_n\), these conditions are satisfied here, so their Corollary 1.2 gives 
\begin{align}
    \lambda(G)
    \ge
    \frac{1-\cos(\pi/n)}{n-1}
    \ge
    \frac{2}{n^3}.
    \label{eq:GZ-gap}
\end{align}

Recall from Lemma~\ref{lem:stationary} and Remark~\ref{rem:gladiator} that both \(P\) and \(G\) are reversible with respect to the same distribution
\(Q\). Therefore, noting that $p_a+p_b\ge 2n^{-c}$, we get
\begin{align*}
    \mathcal E_P(f)&\ge 2(n-1)n^{-c}\mathcal E_G(f)\ge n^{1-c}\mathcal E_G(f)
\end{align*}
for every function \(f:S_n\to\mathbb{R}\), and so
\begin{align*}
    \lambda(P)\ge n^{1-c}\lambda(G)\ge 2n^{-2-c}.
\end{align*}

To convert this into a bound on $\gamma(P)$, we also need to rule out
eigenvalues of \(P\) close to \(-1\).  This is easy because the chain has a
self-loop probability at least \(n^{-c}\) at every state: if the first item of
the list is requested, Transposition does nothing. 
Thus we may write
\[
    P=n^{-c} I+(1-n^{-c})R
\]
for some stochastic matrix \(R\).  Since every eigenvalue of a stochastic matrix has
absolute value at most \(1\), and eigenvalues of \(P\) are
real, every eigenvalue of \(P\) is at least \(2n^{-c}-1\).  Hence the
absolute spectral gap of \(P\) satisfies
\begin{align}
    \gamma(P)
    \ge
    \min\left\{\lambda(P),\,2n^{-c}\right\}
    \ge 2n^{-2-c}.
    \label{eq:absolute-gap}
\end{align}

We now lower-bound the smallest stationary probability.  Let
\[
    N:=\binom n2=\frac{n(n-1)}2 .
\]
For every \(\pi\in S_n\),
\[
    \prod_{i=1}^n p_i^{\,n-\pi(i)}
    \ge
    n^{-cN}.
\]
Also \(Z(p)\le n!\), since each summand in \(Z(p)\) is at most \(1\).  Thus
\begin{align}
    Q(\pi)\ge \frac{n^{-cN}}{n!}
    \qquad\text{for all }\pi\in S_n.
    \label{eq:Qmin}
\end{align}

Let the initial state be \(\pi_0\).  Combining
\eqref{eq:spectral-tv-bound}, \eqref{eq:absolute-gap}, and
\eqref{eq:Qmin}, and using \(1-\gamma(P)\le e^{-\gamma(P)}\), we obtain
\begin{align*}
    \|P^t(\pi_0,\cdot)-Q\|_{\mathrm{TV}}
    &\le
    \frac12 Q(\pi_0)^{-1/2}
    \exp\left(-\frac{2 t}{n^{2+c}}\right)\\
    &\le \exp\left(-\frac{2 t}{n^{2+c}} + O(n^2\log n)\right)
\end{align*}
Choosing $t\ge t_0$ for some $t_0=O(n^{c+4}\log n)$ is large enough to give 
\begin{align}
    \|P^t(\pi_0,\cdot)-Q\|_{\mathrm{TV}}
    \le
    \frac1{n^2}.
    \label{eq:tv-nminus2}
\end{align}

Finally we translate total-variation closeness into a cost bound.  The
function \(\cost(\pi)\) always lies in the interval \([1,n]\).
Therefore, for any two distributions \(\mu,\nu\) on \(S_n\),
\[
    \left|
    \mathbb{E}_{\pi\sim\mu}[\cost(\pi)]
    -
    \mathbb{E}_{\pi\sim\nu}[\cost(\pi)]
    \right|
    \le
    n\|\mu-\nu\|_{\mathrm{TV}} .
\]
Applying this with \(\mu=P^t(\pi_0,\cdot)\) and \(\nu=Q\), and using
\eqref{eq:tv-nminus2}, gives
\[
    \mathbb{E}[\cost(\pi_t)]
    \le
    \mathbb{E}_{\pi\sim Q}[\cost(\pi)]
    +
    \frac{1}{n}.
\]
The proof of the stationary bound above in fact gives the stronger estimate
\[
    \mathbb{E}_{\pi\sim Q}[\cost(\pi)]
    \le
    \OPT(p)+\sum_{j=2}^n p_j
    =
    \OPT(p)+1-p_1 \le \OPT(p)+1-\frac{1}{n}.
\]
We conclude that
\[
    \mathbb{E}[\cost(\pi_t)]
    \le
    \OPT(p)+1 .\qedhere
\]
\end{proof}

\end{document}